\documentclass[conference,10pt]{IEEEtran}

\usepackage{cite}
\ifCLASSINFOpdf
   \usepackage[pdftex]{graphicx}
\else
  \usepackage[dvips]{graphicx}
\fi

\usepackage[cmex10]{amsmath}
\usepackage{amssymb}
\usepackage{amsfonts}
\usepackage{amsthm}
\usepackage{array}
\usepackage{dsfont}
\usepackage{setspace}
\usepackage{fixltx2e}
\usepackage{epstopdf}

\newtheorem{lemma}{Lemma}
\newtheorem{remark}{Remark}
\newtheorem{corollary}{Corollary}
\newtheorem{theorem}{Theorem}

\newtheorem{definition}{Definition}
\long\def\symbolfootnote[#1]#2{\begingroup%
\def\thefootnote{\fnsymbol{footnote}}\footnote[#1]{#2}\endgroup}
\allowdisplaybreaks[1]

\IEEEoverridecommandlockouts

\begin{document}

\title{Smart Meter Privacy with an Energy Harvesting Device and Instantaneous Power Constraints }

\author {
  \IEEEauthorblockN{Giulio Giaconi and Deniz G\"{u}nd\"{u}z}

  \IEEEauthorblockA{Imperial College London,  London, UK\\
    {\{g.giaconi, d.gunduz\}}{@imperial.ac.uk}
\and
\IEEEauthorblockN{H. Vincent Poor}
\IEEEauthorblockA{Princeton University, Princeton, NJ \\
poor@princeton.edu}
}
\thanks{This research was supported in part by the U.S. National Science Foundation under Grant CMMI-1435778.}

\author{
\vspace{-0.7cm}}

}

\maketitle

\begin{abstract}
A \textit{smart meter} (SM) periodically measures end-user electricity consumption and reports it to a utility provider (UP). Despite the advantages of SMs, their use leads to serious concerns about consumer privacy. In this paper, SM privacy is studied by considering the presence of an energy harvesting device (EHD) as a means of masking the user's input load. The user can satisfy part or all of his/her energy needs from the EHD, and hence, less information can be leaked to the UP via the SM. The EHD is typically equipped with a rechargeable energy storage device, i.e., a battery, whose instantaneous energy content limits the user's capability in covering his/her energy usage. Privacy is measured by the information leaked about the user's real energy consumption when the UP observes the energy requested from the grid, which the SM reads and reports to the UP. The minimum information leakage rate is characterized as a computable information theoretic single-letter expression when the EHD battery capacity is either infinite or zero. Numerical results are presented for a discrete binary input load to illustrate the potential privacy gains from the existence of a storage device.
\end{abstract}
\IEEEpeerreviewmaketitle

\section{Introduction}

The transition from a conventional power distribution network to a \textit{smart grid} is a necessity, not only to replace a less efficient system, but also to address new environmental and societal challenges. Smart grids provide many advantages for energy generation, transmission, distribution and consumption, and also allow the integration of renewable energy sources into the power network \cite{Mo:2012}. A key element of a smart grid is the \textit{smart meter} (SM), which records minutely the electricity consumption of a user, thereby permitting accurate estimation and control of the smart grid's behavior.

However, SMs may have unintended consequences regarding the privacy of users. The information collected and distributed inside the smart grid, namely the \textit{load profile} of the users, can potentially be used for other purposes, thereby raising the question of data security and privacy. From the load consumption profiles, i.e., the series of energy usage values that are regularly collected from a user, nonintrusive appliance load monitoring (NALM) techniques can be used to identify the user's habits, such as whether he/she is at home, the equipment he/she is using, and even the TV channel he/she is watching \cite{Rouf:2012}.


Privacy of SM data has recently become a topic of considerable research interest and many solutions have been proposed. In \cite{Costas:2010} the authors designed an anonymization system by introducing two separate SM identities (IDs), and only the low frequency ID is released to the utility provider (UP). Bohli \textit{et al.} \cite{Bohli:2010} considered the aggregation of information to cover individual energy consumption information. Other approaches are to mask the data by obfuscating it, as in \cite{Kim:2011} and \cite{Sankar:2013}, or to use homomorphic encryption, as in \cite{Garcia:2010}.

However, we note that the above techniques, all based on the manipulation of the SM readings, suffer from a further privacy risk. The energy consumed by a user is provided directly from the grid, which is fully controlled by the UP, and hence, the UP can easily embed additional meters to monitor the energy requested by a household or a business, in order to receive accurate information about the energy consumption without relying on SM readings. To cope with this issue an alternative approach is to directly modify the actual energy consumption profile of the user, rather than simply distorting the data sent to the UP. This can be done, for example, by using an energy storage device, as in \cite{Kalogridis:2010}, \cite{Varodayan:2011} and \cite{Tan:2013}, or by introducing an alternative energy source (AES), as proposed in \cite{Tan:2013}, \cite{Gunduz:2013}, \cite{Gomez:2013} and \cite{Gomez:2015}. We adopt the latter approach, and focus on providing privacy by using an energy harvesting device (EHD) with a rechargeable energy storage component, which we duly call a battery. A similar model, studied in \cite{Gomez:2015}, imposes only an average power constraint on the energy the user can receive from the AES, which can be a microgrid, capable of providing any amount of energy at each time instant. The average power constraint may be imposed to limit the cost of the energy received from the AES. However, if the AES is an EHD with a battery, it may not be possible to control the amount of energy harvested at each time instant. Hence, since the energy available in the battery imposes limitations on the capability of the energy management unit (EMU) to reduce information leakage to the UP, it is important to take into account the size and state of the battery when designing an energy management policy.

We study the minimum amount of user's energy consumption information leaked to the UP, under instantaneous power constraints. We provide single-letter expressions for the minimum information leakage rate when the EHD battery capacity is either infinite or zero. In the case of discrete input load distributions, the minimization problem can be shown to be convex, and hence, it can be solved efficiently by numerical methods, such as the Blahut-Arimoto algorithm. Finally, we provide simulation results for a binary input load distribution.

The remainder of the paper is organized as follows. In Section \ref{sec:SystemModel} we introduce the system model, and in Sections \ref{sec:infinite} and \ref{sec:zero} we consider the case of infinite and zero battery capacity, respectively. In Section \ref{sec:binary} we simulate the system by considering a discrete binary input load distribution. Finally, conclusions are drawn in Section \ref{sec:conclusion}.

\section{System Model}\label{sec:SystemModel}

\begin{figure}[t]
\centering
\includegraphics[width=1\columnwidth]{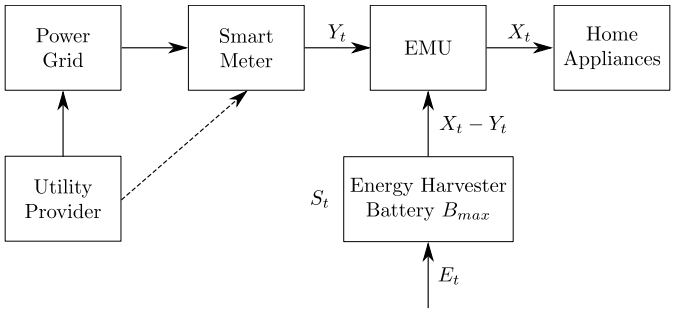}
\caption{The system model. $X_t$, $Y_t$ and $E_t$ are the user input load, the output load, and the energy harvested by the EHD at time $t$, respectively.}
\label{fig:SystemModel}
\end{figure}


We consider the discrete time system depicted in Fig. \ref{fig:SystemModel}, where each discrete time represents a time slot. The input load $X_t \in \mathcal{X}$ is the energy requested by the user in time slot $t$, while $Y_t \in \mathcal{Y}$ is the energy received from the UP, i.e., the output load. We assume that the entries of the input load sequence, $\{X_t\}_{t=1}^{\infty}$, are independent and identically distributed (i.i.d.) with distribution $p_X$. In time slot $t-1$, the EHD harvests $E_t\in \mathcal{E}$ units of energy from the environment, which become available to the EMU at the beginning of time slot $t$. The unused energy is then stored in a battery, whose maximum capacity is $B_{max}$. The harvested energy sequence, $\{E_t\}_{t=1}^{\infty}$, is modelled as having i.i.d. entries with distribution $p_E$. We define the average value of the harvested energy in each time slot as $\bar{P}_E =\mathbb{E}[E]$.

The EMU satisfies the user's energy demand in all time slots from either the UP or the EHD. We do not allow extra energy to be drawn from the grid and then wasted. This could provide additional privacy to the user, albeit at a significantly higher energy cost. Also, the user cannot get energy from the grid and store it in the EHD battery. Hence, we have $Y_t \leq X_t$, and $X_t - Y_t$ denotes the amount of energy the EMU gets from the EHD in time slot $t$. Let $S_t\in [0, B_{max}]$ denote the state of the battery at the beginning of time slot $t$. We have
\begin{align}\label{battery_constraint}
	S_{t+1} = \max \{\min \{S_{t} + E_t - (X_t - Y_t), B_{max} \}, 0 \},
\end{align}
where $S_{t+1}$ denotes the state of the battery at the beginning of time slot $t+1$.


%


By measuring the privacy of a user, it is possible to compare effectively different privacy-preserving strategies and use the best one. However, coming up with a universal privacy measure is elusive. Previous works have explored various measures. Kalogridis \textit{et al.} \cite{Kalogridis:2010} considered the Kullback-Leibler divergence between the input and output loads, cluster classification for the input load and regression analysis. The authors in \cite{McLaughlin:2011} computed the number of features, i.e., the number of events that indicate whether a device is switched on or off, and the empirical entropy, considered as an upper bound on the information extractable via a NALM procedure. The studies in \cite{Varodayan:2011}, \cite{Tan:2013} and \cite{Gunduz:2013} considered \textit{mutual information} between the input and output loads as a measure of information leakage. Similarly, we also minimize the mutual information rate between $X^n$ and $Y^n$. We note that mutual information as a privacy measure is not limited by technology implementations or the complexity of the NALM procedure, and allows us to provide provable bounds on the achievable privacy.

%

We aim at designing energy management policies that will decide on the amount of energy to request from the UP in every time slot $t$, given the input load up to time $t$, $X^t$, the energy requested from the UP up to time $t-1$, $Y^{t-1}$, and the state sequence of the battery $S^{t}$. Our goal is to minimize the information leakage rate, i.e., the average of the mutual information between the input and output load sequences.


\subsection{Definitions}

\begin{definition}
A length-$n$ energy management policy is composed of, possibly random, power allocation functions
\begin{equation}
f_t:\mathcal{X}^t \times \mathcal{Y}^{t-1} \times \mathcal{S}^{t} \rightarrow \mathcal{Y},
\end{equation}
for $t=1,...,n$, such that $0 \leq Y_t \leq X_t$ and the battery state evolves according to (\ref{battery_constraint}).
\end{definition}

\begin{definition}
The privacy achieved by a length-$n$ energy management policy is measured via the \textit{information leakage rate} defined as
\begin{equation}
I_n= \frac{1}{n} I \left( X^n; Y^n \right).
\end{equation}
\end{definition}

\begin{definition}
An information leakage rate $I$ is said to be \textit{achievable} if there exists a sequence of energy management policies such that $\lim_{n \to \infty}I_n \leq I$.
\end{definition}


\begin{definition}
For an SM system with given $p_X, p_E$ and $B_{max}$, the  \textit{minimum privacy leakage rate} $\mathcal{I}$ is defined as the infimum of the achievable information leakage rates.
\end{definition}

In \cite{Gomez:2015} the SM privacy problem is studied with the assumption that the EMU is only constrained by the average and peak values of the energy it can request from the AES. Note that this is a more relaxed problem compared to evaluating the minimum information leakage rate in the presence of a battery. The \textit{privacy-power function}, $\mathcal{I}(\bar{P}, \hat{P})$ is defined in \cite{Gunduz:2013} as the minimum information leakage rate that can be achieved when the energy management policy satisfies $E[\frac{1}{n} \sum_{t=1}^n (X_t - Y_t)] \leq \bar{P}$ and $0 \leq X_t - Y_t \leq \hat{P}$.

\begin{theorem}\label{th:average_peak}
\cite[Theorem 1]{Gunduz:2013} The privacy-power function $\mathcal{I}( \bar{P}, \hat{P})$ for an i.i.d. input load vector $X$ with distribution $p_{X}(x)$, when the average and peak values of the power provided by the AES are limited by $\bar{P}$ and $\hat{P}$, respectively, is given by
\begin{equation}\label{expr_privacy_power}
\mathcal{I}(\bar{P}, \hat{P}) = \inf_{\substack{p_{Y|X}(y|x): 0\leq \mathbb{E}[X-Y] \leq \bar{P} \\ 0 \leq X-Y \leq \hat{P}}} I\left(X;Y\right).
\end{equation}
\end{theorem}

In the following, we compute the minimum information leakage rate for the infinite and zero battery scenarios.

\section{EHD with an Infinite Battery}\label{sec:infinite}

In this scenario we set $B_{max}=\infty$. This is an extreme situation that may occur if we consider particularly large battery sizes compared to the energy harvesting rate $\bar{P}_E$. When $B_{max}=\infty$, the constraints on the energy management policy are relaxed. In each time slot the EMU is limited by the available energy in the battery, which is simply the difference between the total harvested energy up to a certain time and the total energy that has been requested from the EHD up to that time. We have the following cumulative energy constraints:
\begin{equation} \label{eq:constr}
 	\sum_{t=1}^n (X_t - Y_t) \leq \sum_{t=1}^n E_t, \qquad \forall n.
\end{equation}




The following theorem states that the minimum information leakage rate when $B_{max} = \infty$ is equivalent to the average power-constrained scenario; that is, the cumulative constraints on the EMU policy do not reduce the achievable privacy compared to the average power-constrained case.

\begin{theorem}\label{th:intro}
If $B_{max}=\infty$, the minimum information leakage rate $\mathcal{I}_{\infty}$ for an i.i.d. input load $X$ with distribution $p_{X}$, and an energy harvesting process with average power $\bar{P}_E$, is given by
\begin{equation}
\mathcal{I}_{\infty} \triangleq \mathcal{I}(\bar{P}_E, \infty).
\end{equation}
\end{theorem}

It is shown in \cite{Gomez:2013} that, when the input load alphabet $\mathcal{X}$ is discrete, the output alphabet load alphabet $\mathcal{Y}$ can be restricted to $\mathcal{Y} = \mathcal{X}$ without loss of optimality. Given this restriction and the convexity of the privacy-power function, $I_{\infty}$ can be efficiently evaluated.

It is easy to see that $\mathcal{I}(\bar{P}_E, \infty)$ is a lower bound on the minimum information leakage rate under battery constraints, since this is equivalent to replacing the cumulative energy constraints with an average one. Next, we prove the achievability part of the theorem, and present two distinct schemes that achieve $I_{\infty}$.

\subsection{Store-and-Hide Scheme}

 As the name suggests, the \textit{store-and-hide scheme} consists of an initial phase, in which all the energy requests of the user are satisfied from the grid, while the harvested energy is stored in the EHD battery, and a second phase, during which the EMU starts using a stochastic policy to decide how much energy it will request from the EHD. The conditional probability specified in Theorem \ref{th:intro} is used in the second phase.



Consider $n$ time slots. In the first $h(n)$ time slots, called the \textit{storage phase}, no privacy is achieved because all the energy requested by the user is taken from the grid. In the remaining $n-h(n)$ time slots, called the \textit{hiding phase}, the user demand is satisfied by taking energy from both the grid and the battery, thus masking the real consumption. We assume that $h(n)\in o(n)$, with $ \lim_{n \to \infty} h(n)= \infty $, and $ \lim_{n \to \infty} n-h(n)= \infty $.

It is noteworthy that no information about the recharge process of the battery is required, and all the EMU needs to know is the average power of the harvesting process, $\bar{P}_E$. It is then possible to show that the store-and-hide scheme satisfies the constraints in (\ref{eq:constr}) with probability arbitrarily close to $1$ as long as $\mathbb{E}[X-Y] < \bar{P}_E$. In short, we show that for any $\epsilon > 0$ and sufficiently large $n$, the following relation holds:
\begin{equation}
\Pr  \left( \bigcup_{k=1}^{n} \left\{ \sum_{t=1}^{k} E_t < \sum_{t=1}^{k} X_t - Y_t \right\} \right)\le \epsilon.
\label{eq:cond1}
\end{equation}

This condition affirms that, for arbitrarily large $n$, the probability that, at any time slot $k$, the amount of energy stored in the battery is less than the energy required by the EMU will be arbitrarily small. The proof, analogous to that of \cite[Lemma~1]{Ozel:2012}, and omitted due to space limitations, requires the application of the weak and the strong law of large numbers along with the tail behaviors of sums of i.i.d. random variables.

\subsection{Best-Effort Scheme}

In the \textit{best-effort scheme} the EMU does not wait to request energy from the EHD; instead, it follows the same stochastic policy as in the hiding phase of the store-and-hide scheme, and whenever the requested energy is not available in the EHD battery, all the input load is satisfied from the grid.


In each time slot $t$, the EMU, based on the instantaneous input load $X_t$, decides on the portion of the input load to be received from the grid, $Y_t$, using the conditional probability distribution $p_{Y|X}$, which satisfies $\mathbb{E}[X-Y] < \bar{P}_E$. Then, if there is enough energy in the battery to satisfy the requested energy from the EHD, i.e., if the condition $S_{t} + E_{t}\geq X_t-Y_t$ is satisfied, then the EMU gets $X_t - Y_t$ units of energy from the EHD battery, and $Y_t$ from the grid; otherwise, all the input load is satisfied directly from the grid, i.e., $Y_t = X_t$. The energy in the EHD is thus updated according to the following relation:
\begin{equation}
S_{t+1}=S_{t}+E_{t}-(X_{t}-Y_{t}) \cdot \textbf{1}(S_{t}+E_{t}\geq X_{t}-Y_{t}),
\end{equation}
where $\textbf{1}(S_{t}+E_{t}\geq X_{t}-Y_{t})=1$ if $S_{t}+E_{t}\geq X_{t}-Y_{t}$, and $0$ otherwise.

Note that when there is not enough energy in the battery, the energy request of the user is satisfied completely by the UP, leading to the maximum information leakage. To find the overall information leakage rate, we need to show that these events happen only for a limited number of time slots. This is proven in the following lemma.

\begin{lemma} \label{lem2}
In the best-effort scheme, if $\mathbb{E}[X-Y] < \bar{P}_E$, then the condition $S_{t}+E_{t}< X_t - Y_t$ holds only for finitely many time slots as $n$ grows to infinity.
\end{lemma}

The proof, omitted due to space limitations, is an application of the strong law of large numbers. It follows from Lemma \ref{lem2} that the leakage of full information takes place only at finitely many time slots as $n \rightarrow \infty$. Thus, the information leakage rate is equivalent to that achieved under an average power constraint.

\begin{remark}
The two energy management policies described above achieve the same privacy performance even though they have some conceptual differences. It is important to remark that the initial assumption is to have an i.i.d. energy harvesting process. However, this condition is typically not satisfied. For instance, if we consider a solar panel as an EHD, the recharge process is clearly not i.i.d and varies greatly from day to night, and also over seasons. Thus, an optimal management strategy similar to the one described in \cite{Ozel:2012} in the dual problem of capacity determination, can be derived. In that paper, the authors divide the energy arrivals and expenditures into periods, consider different average harvesting rates for every period, and finally maximize the average throughput over all the periods.
\end{remark}


\section{EHD without a Battery} \label{sec:zero}

Here, we focus on the case in which the EHD has no battery for energy storage, i.e., $B_{max}=0$. The energy harvested in time slot $t-1$ and available at the beginning of time slot $t$, $E_t$, can be considered as the state information. Given $E_t$ and the input load $X_t$, the EMU decides on the amount of energy to use from the grid and from the EHD. Thus, this is an SM system with a stochastic peak power constraint on the energy that the EMU can obtain from the EHD. This constraint, represented by $E_t$, is causally known by the EMU, while the UP may or may not know its realization.




In each time slot $t$ the energy requested from the EHD, $X_t-Y_t$, is limited by the energy harvested in time slot $t-1$, $E_t$, i.e.,
\begin{equation}\label{eq:peakconstraints}
0 \leq X_t-Y_t\leq E_t, \qquad t=1,...,n.
\end{equation}

Note that, as opposed to the infinite battery scenario in (\ref{eq:constr}), here the past has no influence on the energy constraints, since there is no battery, and thus no memory in the system.


We first consider the minimum leakage rate when the EHD harvests a constant and fixed amount of energy in every time slot, i.e., $\mathcal{E} = e$. We remark that $e$ is not a random variable, and is known by both the EMU and the UP. The privacy-power function is thus obtained by considering only a peak power constraint on the amount of energy the EMU can request from the harvester. We denote the corresponding minimum information leakage rate by $\mathcal{I}(e)$, which can be obtained as a special case of Theorem \ref{th:average_peak}. We state this in the following corollary as it will be used later.

\begin{corollary}\label{c:ppfsingle}
If $B_{max}=0$ and $\mathcal{E} = e$, the privacy-power function $\mathcal{I}(e)$ for an i.i.d. input load $X$ is given by
\begin{equation}\label{eq:energyConstant}
\mathcal{I}(e) = \mathcal{I}(e,e).
\end{equation}
\end{corollary}

Corollary \ref{c:ppfsingle} follows directly from Theorem \ref{th:average_peak}, by considering only the peak power constraint.

When the harvested energy $E$ changes randomly over time, this corresponds to an SM privacy problem with a random state. In this case, we can identify two different scenarios based on the information available at the UP regarding the state of the system.

\subsection{EHD State Known only by the EMU}

In this scenario, the instantaneous state of the harvested energy is assumed to be known only at the EMU. However, we still assume that the UP knows the probability distribution of the states, i.e.,  $p_E$. This scenario can occur, for example, when the EHD models a connection to a secondary grid or a microgrid, and the primary UP cannot have any information about the exact amount of energy provided by this AES. Fig. \ref{fig:causalTx} depicts this scenario.

\begin{figure}[t]
\centering
\includegraphics[width=0.8\columnwidth]{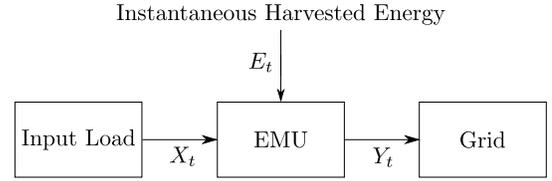}
\caption{EHD state information $E_t$ available only to the EMU.}
\label{fig:causalTx}
\end{figure}


%

\begin{theorem}\label{th:IT}
If $B_{max}=0$, and the energy available at the EHD in each time slot, E, is i.i.d. with distribution $p_E$, then the minimum information leakage rate $\mathcal{I}_0$ is given by
\begin{equation}\label{expr_zero}
\mathcal{I}_0 \triangleq \inf_{\substack{p_{Y|X,E}(y|x,e): 0\leq X-Y \leq E}} I\left(X;Y\right).
\end{equation}
\end{theorem}

\begin{proof}
\textit{Achievability} We consider a conditional probability distribution $p_{Y| X, E}(y|x, e)$ that satisfies the conditions of Theorem \ref{th:IT} and we generate each $Y_t$ independently using $p_{Y| X,E}(y_t| x_t, e_t)$. The mutual information leakage rate is therefore given by $I(X;Y)$ whereas the peak power constraint in Theorem \ref{th:IT} are trivially satisfied.

\textit{Converse}. We assume that there is a series of power allocation functions that satisfy the peak power constraints. The information leakage rate of the resulting output load series satisfies the following chain of inequalities:
\begin{subequations}
\begin{align}
\label{eq:ppwT1}
 & \frac{1}{n}  I(X^{n};Y^{n})  = \frac{1}{n}\left[H(X^{n})-H(X^{n}|Y^{n} )\right] \\
 \label{eq:ppwT2}
 & =    \frac{1}{n} \sum_{t=1}^n [ H(X_t)-H(X_t|X^{t-1}Y^n)] \\
  \label{eq:ppwT3}
 & \ge 	\frac{1}{n} \sum_{t=1}^n [H(X_t)-H(X_t|Y_t)] \\
  \label{eq:ppwT4}
 & = 	\frac{1}{n} \sum_{t=1}^n I\left(X_t;Y_t\right) \\
 \label{eq:ppwT5}
 & \ge 	\frac{1}{n} \sum_{t=1}^n \inf_{\substack{p_{Y|X,E}(y|x,e): 0\leq X-Y \leq E}}   I\left(X;Y\right) \\
  \label{eq:ppwT6}
 & = 	 \inf_{\substack{p_{Y|X,E}(y|x,e): 0\leq X-Y \leq E}} I\left(X;Y\right) \\
  \label{eq:ppwT7}
 & = 	\mathcal{I}_0 ,
\end{align}
\end{subequations}
where (\ref{eq:ppwT3}) follows as conditioning reduces entropy; (\ref{eq:ppwT5}) follows because whatever strategy the EMU adopts, $I\left(X_t;Y_t\right)$ can never be smaller than the minimum of the mutual information over all conditional probability distribution functions $p_{Y|X,E}(y|x,e)$ that satisfy the peak power constraint; and (\ref{eq:ppwT7}) follows from the definition of $\mathcal{I}_0$.
\end{proof}

Theorem \ref{th:IT} suggests that the minimum information leakage rate in the absence of a storage device at the EHD can be characterized as the solution of a single-letter information theoretic optimization problem. Similarly to the privacy-power function with average and peak power constraints in (\ref{expr_privacy_power}), when the input load is discrete, the output load can be limited to the input load without loss of optimality, and the infimum becomes a minimum, which can be computed efficiently.

\subsection{EHD State Information Known at the UP}

\begin{figure}
\centering
\includegraphics[width=0.8\columnwidth]{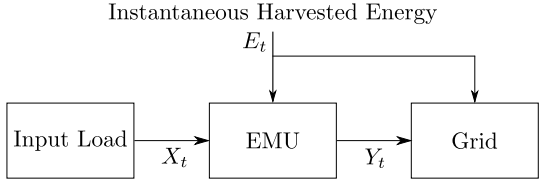}
\caption{State information $E_t$ available to the EMU and the UP.}
\label{fig:causalTxRx}
\end{figure}

We now assume that the state $E_t$ of the EHD is known by the UP as well as the EMU (Fig. \ref{fig:causalTxRx}). This is a worst-case situation and we expect that the amount of leaked information in this case is greater than or equal to that of the previous scenario, in which only the EMU knows the current state of the harvested energy. This situation can occur, for example, if the EHD is a solar panel, and the UP can estimate the harvested energy accurately from the weather forecast of the area and the specifications of the solar panel.



The privacy-power function when the harvested energy is i.i.d. with distribution $p_E$ and $B_{max}= 0$, is provided in the next theorem. We denote this quantity as $\bar{\mathcal{I}}_0$.

\begin{theorem}\label{th:ITR}
If $B_{max}=0$, the input load is i.i.d. with distribution $p_X$ and the state of the harvested energy is available at the UP, then the minimum information leakage rate $\bar{\mathcal{I}}_0$ is given by
\begin{equation}\label{leakage_TR}
\bar{\mathcal{I}}_0 = \inf_{\substack{p_{Y|X,E}(y|x,e): 0\leq X-Y \leq E}} I\left(X;Y|E \right) = \mathbb{E}_E [\mathcal{I}(E)].
\end{equation}
\end{theorem}

\begin{proof}
\text{Achievability} follows trivially, as in Theorem \ref{th:IT}.

\textit{Converse}. We assume that there is a series of power allocation functions that satisfy the peak power constraints. We then want to show that the information leakage rate of the resulting output load series is lower bounded by the value given in (\ref{leakage_TR}). The following chain of inequalities applies:
\begin{subequations}
\allowdisplaybreaks
\begin{align}
\allowdisplaybreaks
 \frac{1}{n}  I(X^{n};&Y^{n}|E^n) = \frac{1}{n}\left[H(X^{n}|E^n)-H(X^{n}|Y^{n}, E^n )\right] \nonumber \\
 \label{eq:ppwTR2}
 & =    \frac{1}{n} \sum_{t=1}^n [H(X_t|E_t)-H(X_t|X^{t-1} Y^n E^n)] \\
  \label{eq:ppwTR3}
 & \ge 	\frac{1}{n} \sum_{t=1}^n [H(X_t|E_t)-H(X_t|Y_t E_t)] \\
  \label{eq:ppwTR4}
 & = 	\frac{1}{n} \sum_{t=1}^n  \sum_{k=1}^{|\mathcal{E}|} p_{e_k}   I\left(X_t;Y_t|E_t=e_k\right) \\
  \label{eq:ppwTR5}
 & \ge 	\frac{1}{n} \sum_{t=1}^n  \sum_{k=1}^{|\mathcal{E}|} p_{e_k} \mathcal{I}\left(e_k\right) \\
  \label{eq:ppwTR7}
 & = \sum_{k=1}^{|\mathcal{E}|} p_{e_k} \mathcal{I}\left(e_k\right)  \\
  \label{eq:ppwTR8}
 & = 	\mathbb{E}_E \left[ \mathcal{I}(E) \right],
\end{align}
\end{subequations}
where (\ref{eq:ppwTR3}) follows as conditioning reduces entropy; (\ref{eq:ppwTR4}) follows from considering all the states $e_k$; and (\ref{eq:ppwTR5}) follows from the definition of $\mathcal{I}(e)$ in (\ref{eq:energyConstant}).
\end{proof}

Theorem \ref{th:ITR} suggests that the minimum information leakage rate $\bar{\mathcal{I}}_0$ is the expected value of the privacy-power function with peak power constraint over the distribution of the states.


We have the following inequality between the privacy power function $\bar{\mathcal{I}}(e)$ and $\mathcal{I}(e, \infty)$, where the latter is the minimum information leakage rate with an average power constraint $e$:
\begin{equation}
\bar{\mathcal{I}}(e) \geq \mathcal{I} \left(e, \infty \right).
\label{eq:relat}
\end{equation}
Similarly, the following inequality also holds:
\begin{equation}
\bar{\mathcal{I}}_0 = \mathbb{E}_E \left[ \mathcal{I}(E) \right] \geq \mathcal{I} \left( \bar{P}_E, \infty \right) = \mathcal{I}_{\infty}.
\label{eq:relat2}
\end{equation}

(\ref{eq:relat2}) states that the minimum information leakage rate when the EHD does not have a battery, and the harvested energy state is random but known at the UP, is greater than or equal to that which can be achieved with an infinite battery.


From the chain rule of mutual information, we have
\begin{subequations}
\begin{align}\label{eq:compar1}
 I(X;YE)  &= I(X;E) + I(X;Y|E) = I(X;Y|E), \\
 I(X;YE)  &= I(X;Y) + I(X;E|Y), \label{eq:compar2}
\end{align}
\end{subequations}
where (\ref{eq:compar1}) follows since $X$ and $E$ are independent of each other. We have $I(X;E|Y)\geq 0$, from the non-negativity of mutual information; and thus, $I(X;Y) \leq I(X;Y|E)$. From Theorems \ref{th:IT} and \ref{th:ITR}, we have $\mathcal{I}_0 \leq \bar{\mathcal{I}}_0$.

In the next section we provide numerical results for both the infinite and zero battery scenarios studied in this paper.


\section{Binary Input Load} \label{sec:binary}

We compare the minimum information leakage rates for the infinite and zero battery scenarios in the binary input load case, i.e., $\mathcal{X}= \{0,1\}$. $X$ follows an independent Bernoulli distribution with $\Pr\{X=1\}=q_x$. We also consider a binary energy harvesting process, i.e., $\mathcal{E}= \{0,1\}$, where $\Pr\{E_t=1\} = p_e$ denotes the probability that the EHD harvests one unit of energy.

\begin{figure}[t]
\centering
\includegraphics[width=0.8\columnwidth]{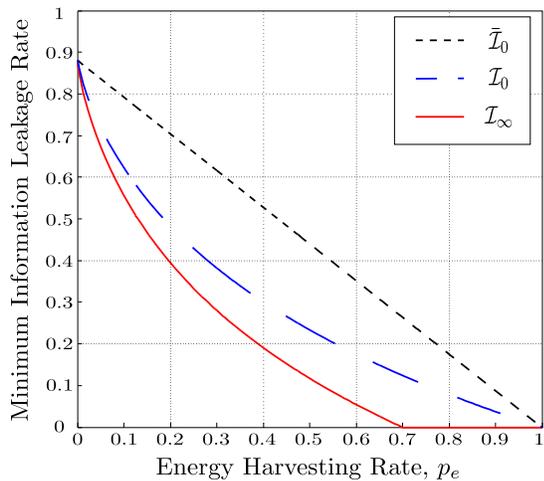}
\caption{Comparison of the minimum information leakage rate in the infinite and zero battery scenarios with binary input load and a binary energy harvesting process. We set $q_x=0.7$, and plot the minimum information leakage rate with respect to $p_e$.}
\label{fig:BinaryZeroInfinite}
\end{figure}

If $B= \infty$, the minimum information leakage rate is given by
\begin{align*}
\mathcal{I}(p_e, \infty) =& p_e \log_2 p_e - q_x \log_2 q_x \\ &- (1- q_x + p_e) \log_2 (1- q_x + p_e),
\end{align*}
if $p_e \leq q_x$, and $\mathcal{I}(p_e, \infty) = 0$ otherwise.

When $B_{max}=0$, there are two scenarios. The first is that in which the state information is known only by the EMU. We need to identify the optimal conditional distribution $p_{Y|X,E}$ with $X-Y \leq E$, that minimizes $I(X;Y)$. As mentioned earlier we can set $\mathcal{Y} = \{0,1\}$ without loss of optimality. In this scenario, another parameter of interest is the probability of using the energy available at the EHD when the input load is $1$ and there is harvested energy, i.e., $E=1$. It is possible to show that the information leakage rate is minimized when we always use the available energy when there is non-zero input load. In this case the  minimum information leakage rate is found to be $h(1 - q_x + p_e \cdot q_x) - q_x \cdot h(p_e)$, where $h(\cdot)$ is the binary entropy function defined as $h(p) = - p \log_2 p - (1-p) \log_2(1-p)$.

Finally, we consider the zero-battery case when the state is known also by the UP. When the peak power constraint is $e=1$, no information is leaked, whereas if $e=0$, the input load is known perfectly by the UP, leading to a leakage of $H(X)$. Hence, the minimum information leakage rate when the state information is known by the UP is given by $(1-p_e)\cdot h(q_x)$.

In Fig. \ref{fig:BinaryZeroInfinite} we plot the minimum information leakage rate for the three scenarios as a function of $p_e$. As expected, the least information leakage rate is achieved when $B_{max}=\infty$, since this condition is equivalent to having only an average power constraint. The worst scenario is when $B=0$ and the receiver knows the state, denoted by $\bar{\mathcal{I}}_0$. The information leakage rate reduces significantly if the state is not known by the UP. The difference between the solid and the long-dashed curves in the figure illustrates the benefit of having a battery at the EHD. We can see that the gain from the battery is much higher when the harvesting rate is higher, i.e., when $p_e$ is high. This is expected since when $p_e$ is low, there is less energy to be stored for future time slots.


\section{Conclusions} \label{sec:conclusion}

We have studied the information leakage in an SM system by considering an EHD with an infinite or a zero capacity battery. In both cases, we have provided single letter expressions for evaluating the minimum information leakage rate, which indicates the optimal level of information theoretic privacy that can be achieved. Considering the zero battery scenario, we have also studied the information leakage rate when the UP knows the exact amount of harvested energy in each time slot. Finally, we have provided some simulation results for the case of a binary input load. In our future work we plan to study the minimum information leakage rate with a finite capacity battery, which is a challenging problem due to the memory introduced into the system by the existence of the battery.

\bibliographystyle{ieeetran}
\bibliography{refICC2015}

\end{document}